\documentclass[11pt,english,table]{article}
\usepackage{lmodern}

\usepackage[T1]{fontenc}
\usepackage[latin9]{inputenc}
\usepackage{geometry}
\usepackage{amsmath}
\usepackage{amsthm}
\usepackage{amssymb}
\usepackage{setspace}
\usepackage[authoryear]{natbib}
\makeatletter
\theoremstyle{remark}
\newtheorem*{rem*}{\protect\remarkname}
\theoremstyle{plain}
\newtheorem{thm}{\protect\theoremname}
\theoremstyle{plain}
\newtheorem{lem}{\protect\lemmaname}

\setcitestyle{round}
\usepackage{lscape}
\usepackage{longtable}
\usepackage{xcolor}
\usepackage{rotating}

\usepackage{array}
\usepackage{tabularx}\usepackage{multirow}\usepackage{booktabs}
\usepackage{ragged2e}\newcolumntype{C}[1]{>{\centering\arraybackslash}p{#1}}
\usepackage{rotfloat}
\newcolumntype{J}[1]{>{\justify\arraybackslash}p{#1}}
\newcolumntype{R}[1]{>{\RaggedLeft\arraybackslash}p{#1}}
\newcolumntype{Q}[1]{>{\columncolor{Gray}\RaggedLeft\arraybackslash}p{#1}}
\newcolumntype{L}[1]{>{\RaggedRight\arraybackslash}p{#1}}
\newcolumntype{G}{@{\extracolsep{0.5cm}}l@{\extracolsep{0pt}}}%
\newcolumntype{P}[1]{>{\centering\arraybackslash}p{#1}}
\newcolumntype{Y}{>{\centering\arraybackslash}X}
\newcommand{\nhphantom}[1]{\sbox0{#1}\hspace{-\the\wd0}} 

\AtBeginDocument{

}

\usepackage[colorlinks,
            linkcolor=blue,
            anchorcolor=blue,
            citecolor=blue]{hyperref}

\usepackage{babel}

\makeatother

\usepackage{babel}
\providecommand{\lemmaname}{Lemma}
\providecommand{\remarkname}{Remark}
\providecommand{\theoremname}{Theorem}

\begin{document}
\title{Characterizing Correlation Matrices that Admit a Clustered Factor
Representation\emph{\normalsize{}\medskip{}
}}
\author{\textbf{Chen Tong}$^{\ddagger}$\thanks{Chen Tong acknowledges financial support from the Ministry of Education
of China, Humanities and Social Sciences Youth Fund (22YJC790117).}\textbf{ } and\textbf{ Peter Reinhard Hansen}$^{\mathsection}$\bigskip{}
 \\
\\
 {\normalsize{}$^{\ddagger}$}\emph{\normalsize{}Department of Finance, School
of Economics, Xiamen University}{\normalsize{} }\\
 {\normalsize{}$^{\mathsection}$}\emph{\normalsize{}University of North Carolina
\& Copenhagen Business School}{\normalsize{} }\emph{\normalsize{}\medskip{}
 }}
\date{\emph{\normalsize{}\today}}
\maketitle
\begin{abstract}
The Clustered Factor (CF) model induces a block structure on the correlation
matrix and is commonly used to parameterize correlation matrices.
Our results reveal that the CF model imposes superfluous restrictions
on the correlation matrix. This can be avoided by a different parametrization,
involving the logarithmic transformation of the block correlation
matrix.

\bigskip{}
\end{abstract}
\textit{\small{}{\noindent}Keywords:}{\small{} Block correlation matrix,
copula, clustering, factor model.}{\small\par}

\noindent \textit{\small{}{\noindent}JEL Classification:}{\small{}
C38 }{\small\par}

\clearpage{}

\section{Introduction}

In empirical models involving high-dimensional covariance matrices,
it is often beneficial to introduce a parsimonious structure to mitigate
issues stemming from overfitting. A commonly used approach is the
Clustered Factor (CF) model, which is characterized by a linear factor
structure and group clustering, where factor loadings are shared within
each group.

The clustered factor model emerged from high-dimensional copula models
with a multi-factor structure, see e.g \citet{KrupskiiJoe2013} and
\citet{CrealTsay2015}. Clustering became a popular additional structural
component, because it is easy to interpret and makes the model more
parsimonious, see \citet{KrupskiiJoe2015}, \citet[2017b, 2023]{OhPatton2017JBES},
\citet{MannerStarkWied2019}, and \citet{OpschoorLucasBarraVanDick:2021}.\nocite{OhPatton2017}\nocite{OhPatton2023}

It is well known that the CF model induces a block structure on the
correlation matrix (whenever variables are ordered by group assignments).
However, the CF model is not consistent with all block correlation
matrices. It cannot generate negative within-group correlations. However,
this is unlikely to be problematic in empirical applications, because
variables in the same cluster are expected to share traits and be
positively correlated. It is more problematic that the CF model imposes
other restrictions that rules out a class of block correlation matrices,
which seem perfectly reasonable and ordinary from an empirical viewpoint.
Below we fully characterize the class of block correlation matrices
that are coherent with the CF model, and express the superfluous restrictions
the CF model imposes as a testable hypothesis. We then proceed to
highlight an alternative parametrization of block correlation matrices
by \citet{ArchakovHansen:CanonicalBlockMatrix}, which is based on
the matrix logarithm of the correlation matrix of \citet{ArchakovHansen:Correlation}.
This parametrization uniquely parametrizes any non-singular block
correlation matrix, without imposing additional structure. A simple
linear factor structure is also applicable to this parametrization,
if a more parsimonious model is required.

The rest of this paper is organized as follows. In Section \ref{sec:FactorBlock},
we introduce the required notation and the CF model followed by the
theoretical results that characterizes the block correlation matrices
that admit CF model. We express the CF structure as a testable hypothesis,
and detail a testing procedure for selecting the number of factors.
In Section 3, we presents the alternative parametrization by \citet{ArchakovHansen:CanonicalBlockMatrix}
and illustrate it with a simple example. We conclude in Section 4
with a brief summary.

\section{Clustered Factor Models and Block Correlation Matrices\label{sec:FactorBlock}}

Consider an $n$-dimensional random variable, $X\in\mathbb{R}^{n}$,
with non-singular covariance matrix $\Sigma=\mathrm{var}(X)$. Then
the corresponding correlation matrix is given by 
\[
C=\Lambda_{\sigma}^{-1}\Sigma\Lambda_{\sigma}^{-1},
\]
where $\Lambda_{\sigma}=\mathrm{diag}(\sqrt{\Sigma_{11}},\ldots,\sqrt{\Sigma_{nn}})$,
and the vector of standardized variables, 
\[
Z=\Lambda_{\sigma}(X-\mu),\qquad\mu=\mathbb{E}X,
\]
is such that $C=\mathbb{E}(ZZ^{\prime})$.

In the copula literature, it is common to parametrize $C$ with the
clustered factor (CF) model, which is characterized by a linear factor
structure for $Z$ and a partitioning of the variables into clusters/groups.
The factor structure is defined by,
\begin{equation}
Z_{i}=b_{i}^{\prime}f+\varepsilon_{i},\qquad i=1,\ldots,n,\label{eq:FactorModel}
\end{equation}
where $f\in\mathbb{R}^{r}$ is a common vector with $r$ orthogonal
factors, $\mathbb{E}f=0$ and $\mathrm{var}(f)=I_{r}$, and the idiosyncratic
variables $(\varepsilon_{1},\ldots,\varepsilon_{n})$ are mutually
uncorrelated and uncorrelated with the factors, $\mathrm{cov}(f,\varepsilon_{i})=0$
and $\mathrm{cov}(\varepsilon_{i},\varepsilon_{j})=0$ for all $i\neq j$.
The vectors of \emph{factor loadings} are given by $b_{i}={\rm corr}(f,Z_{i})\in\mathbb{R}^{r}$,
for $i=1,\ldots,n$, and it follows that $\mathrm{var}(\varepsilon_{i})=1-b_{i}^{\prime}b_{i}$. 

A partitioning of the $n$ elements into $K<n$ groups adds additional
structure. It is assumed that the factor loadings are common within
each group, such that
\begin{equation}
b_{i}=\beta_{k}\quad\text{for all }i\in{\rm G}_{k},\label{eq:FactorFiniteSet}
\end{equation}
where $({\rm G}_{1},\ldots,{\rm G}_{K})$ is a partition of $\{1,\ldots,n\}$,
with ${\rm G}_{k}$ containing $n_{k}$ elements, $k=1,\ldots,K$,
such that $n=\sum_{k}n_{k}$. 

The CF model implies that the correlation between two variables is
solely determined by their group classification. For $i\in{\rm G}_{k}$
and $j\in{\rm G}_{l}$ with $i\neq j$ we have
\[
C_{ij}={\rm corr}(z_{i},z_{j})=\rho_{kl}\equiv\beta_{k}^{\prime}\beta_{l}.
\]
By rearranging $X_{1},\ldots,X_{n}$ according to their group assignments,
we obtain a block correlation matrix that can be expresses as 
\begin{equation}
\ensuremath{C=\left[\begin{array}{cccc}
C_{[1,1]} & C_{[1,2]} & \cdots & C_{[1,K]}\\
C_{[2,1]} & C_{[2,2]}\\
\vdots &  & \ddots\\
C_{[K,1]} &  &  & C_{[K,K]}
\end{array}\right]},\label{eq:BlockC}
\end{equation}
where $C_{[k,l]}$ is an $n_{k}\times n_{l}$ matrix given by
\[
C_{[k,k]}=\left[\begin{array}{cccc}
1 & \rho_{kk} & \cdots & \rho_{kk}\\
\rho_{kk} & 1 & \ddots\\
\vdots & \ddots & \ddots\\
\rho_{kk} &  &  & 1
\end{array}\right]\ensuremath{\quad}{\rm and}\ensuremath{\quad C_{[k,l]}=\left[\begin{array}{ccc}
\rho_{kl} & \cdots & \rho_{kl}\\
\vdots & \ddots\\
\rho_{kl} &  & \rho_{kl}
\end{array}\right]\quad}{\rm if}\ \ensuremath{k\neq l}.
\]

The advantage of a block correlation matrix, is that its number of
unique correlations is (at most) $d=K\left(K+1\right)/2$, whereas
the number of distinct correlations in an unrestricted $n\times n$
correlation matrix is $n\left(n-1\right)/2$. If $K$ is fixed, then
$d$ does not increase with $n$, which makes it possible to model
high-dimensional correlation matrices with relatively few parameters. 

\subsection{When is the Clustered Factor Structure Applicable}

An interesting question is whether a give block correlation matrix
can be expressed as a CF model. Any block correlation matrix is clearly
coherent with a group partitioning (\ref{eq:FactorFiniteSet}), but
additional structure is needed for it to have the representation in
(\ref{eq:FactorModel}). In some applications, it will be relevant
to know if the factor structure rules out empirically relevant correlation
matrices. It is therefore interesting to characterize the set of block
correlation matrices that are compatible with the CF model. 

For later use, we define the matrix $B=\left(\beta_{1},\beta_{2},\ldots,\beta_{K}\right)^{\prime}\in\mathbb{R}^{K\times r}$,
where $\beta_{k}\in\mathbb{R}^{r\times1}$ is the vector of group-specific
factor loadings, as defined in (\ref{eq:FactorModel}) and (\ref{eq:FactorFiniteSet}).

From a block correlation matrix, $C$, we define the following $K\times K$
matrix
\begin{equation}
A^{*}=\left[\begin{array}{cccc}
\rho_{11} & \rho_{12} & \cdots & \rho_{1K}\\
\rho_{21} & \rho_{22}\\
\vdots &  & \ddots\\
\rho_{K1} &  &  & \rho_{KK}
\end{array}\right].\label{eq:Astar}
\end{equation}
In this paper we focus, without loss of generality, on the case where
$C$ is nonsingular, which implies that $\rho_{kk}<1$ for all within-group
correlations, $k=1,\ldots,K$.
\begin{rem*}
If $n_{k}=1$ for some $k$, then $A^{\ast}$ is not fully defined
by $C$, because there is no within-group correlation in the $k$-th
group, leaving $\rho_{kk}$ undefined. This is the reason that we
treat the case with singleton groups separately in the following Theorem.
\end{rem*}
\begin{thm}
\label{thm:FactorRepresentation}Let $C$ be a nonsingular block correlation
matrix. If $\min_{k}n_{k}\geq2$, then $C$ has a factor representation
if and only if $A^{*}$ is positive semi-definite. If, instead, $n_{k_{1}}=\cdots=n_{k_{m}}=$1
(and $n_{k}\geq2$ for all other), then $C$ has a factor representation
if and only if there exist $\rho_{k_{1}k_{1}},\ldots,\rho_{k_{m}k_{m}}<1$
such that $A^{*}$ is positive semi-definite.
\end{thm}
\begin{proof}
Suppose that $A^{*}$ is positive semi-definite (psd) with $\rho_{kk}<1$
for all $k$, then it can be expressed as $A^{*}=BB^{\prime}$, where
$\beta_{k}^{\prime}\beta_{k}=\rho_{kk}<1$. This shows that $C$ admits
a factor representation. (The matrix of factor loadings, $B$, is
not unique, but one valid choice is $B=Q\Lambda^{1/2}$, where $A^{*}=Q\Lambda Q^{\prime}$
is the eigendecomposition of $A^{\ast}$). That the resulting correlation
matrix, $C$, is nonsingular follows from Lemma \ref{lemma1}.

Now suppose $C$ is a nonsingular block correlation matrix that is
given from CF model with factor loadings $B=(\beta_{1},\ldots,\beta_{K})^{\prime}$.
Then it follows that $\rho_{kl}=\beta_{k}^{\prime}\beta_{l}$, such
that $A^{\ast}=BB^{\prime}$ is psd with $\rho_{kk}=\beta_{k}^{\prime}\beta_{k}<1$
for all $k$. This completes the proof.
\end{proof}
\begin{lem}
\label{lemma1}Suppose that the matrix $A^{*}$ is positive semi-definite
with $\rho_{kk}<1$, for all $k=1,\ldots,K$. Then the corresponding
block correlation matrix, $C$, is positive definite.
\end{lem}
\begin{proof}
From \citet{ArchakovHansen:CanonicalBlockMatrix} it follows that
a block correlation matrix is positive definite if and only if
\[
A=\left[\begin{array}{cccc}
1+\left(n_{1}-1\right)\rho_{11} & \rho_{12}\sqrt{n_{1}n_{2}} & \cdots & \rho_{1K}\sqrt{n_{1}n_{K}}\\
\rho_{21}\sqrt{n_{1}n_{2}} & 1+\left(n_{2}-1\right)\rho_{22}\\
\vdots &  & \ddots\\
\rho_{K1}\sqrt{n_{K}n_{1}} &  &  & 1+\left(n_{K}-1\right)\rho_{KK}
\end{array}\right],
\]
is positive definite with $|\rho_{kk}|<1$, for all $k$. This matrix
can be expressed as $A=\Lambda_{n}^{1/2}A^{*}\Lambda_{n}^{1/2}+\Psi_{1-\rho}$,
where
\[
\Psi_{1-\rho}=\left[\begin{array}{ccc}
1-\rho_{11} & \cdots & 0\\
\vdots & \ddots & \vdots\\
0 & \cdots & 1-\rho_{KK}
\end{array}\right],\qquad\text{and}\qquad\ensuremath{\Lambda_{n}=\left[\begin{array}{ccc}
n_{1} &  & 0\\
 & \ddots\\
0 &  & n_{K}
\end{array}\right]}.
\]
If matrix $A^{*}$ is psd with $\rho_{kk}<1$, then $\Psi_{1-\rho}$
is positive definite and $\Lambda_{n}^{1/2}A^{*}\Lambda_{n}^{1/2}$
is psd. It now follows that $A$ is positive definite because it is
the sum of a positive definite matrix and a positive semidefinite
matrix, and from \citet[theorem 1]{ArchakovHansen:CanonicalBlockMatrix}
it follows that the corresponding block correlation matrix is also
positive definite.\medskip{}
\end{proof}
Although Lemma \ref{lemma1} shows that a psd $A^{*}$ with $\rho_{kk}<1$
implies that $A$ is positive definite, the converse is not true.
We have that $A^{*}=\Lambda_{n}^{-1/2}\left(A-\Psi_{1-\rho}\right)\Lambda_{n}^{-1/2}$,
which is psd if and only if $A-\Psi_{1-\rho}$ is psd. The latter
is not guaranteed. The simplest example is if a within-group correlation
is negative, since the diagonal elements of $A-\Psi_{1-\rho}$ equal
$n_{k}\rho_{kk}$, $k=1,\ldots,K$. Another, less obvious, example
is the following block correlation matrix:
\begin{equation}
C=\left[\begin{array}{ccccccccc}
1 & \bullet & \bullet & \cellcolor{gray!15}\bullet & \cellcolor{gray!15}\bullet & \cellcolor{gray!15}\bullet & \bullet & \bullet & \bullet\\
0.70 & 1 & \bullet & \cellcolor{gray!15}\bullet & \cellcolor{gray!15}\bullet & \cellcolor{gray!15}\bullet & \bullet & \bullet & \bullet\\
0.70 & 0.70 & 1 & \cellcolor{gray!15}\bullet & \cellcolor{gray!15}\bullet & \cellcolor{gray!15}\bullet & \bullet & \bullet & \bullet\\
\cellcolor{gray!15}0.58 & \cellcolor{gray!15}0.58 & \cellcolor{gray!15}0.58 & 1 & \bullet & \bullet & \cellcolor{gray!15}\bullet & \cellcolor{gray!15}\bullet & \cellcolor{gray!15}\bullet\\
\cellcolor{gray!15}0.58 & \cellcolor{gray!15}0.58 & \cellcolor{gray!15}0.58 & 0.63 & 1 & \bullet & \cellcolor{gray!15}\bullet & \cellcolor{gray!15}\bullet & \cellcolor{gray!15}\bullet\\
\cellcolor{gray!15}0.58 & \cellcolor{gray!15}0.58 & \cellcolor{gray!15}0.58 & 0.63 & 0.63 & 1 & \cellcolor{gray!15}\bullet & \cellcolor{gray!15}\bullet & \cellcolor{gray!15}\bullet\\
0.54 & 0.54 & 0.54 & \cellcolor{gray!15}0.19 & \cellcolor{gray!15}0.19 & \cellcolor{gray!15}0.19 & 1 & \bullet & \bullet\\
0.54 & 0.54 & 0.54 & \cellcolor{gray!15}0.19 & \cellcolor{gray!15}0.19 & \cellcolor{gray!15}0.19 & 0.71 & 1 & \bullet\\
0.54 & 0.54 & 0.54 & \cellcolor{gray!15}0.19 & \cellcolor{gray!15}0.19 & \cellcolor{gray!15}0.19 & 0.71 & 0.71 & 1
\end{array}\right],\label{eq:3x3block}
\end{equation}
which is a positive definite, it's smallest eigenvalue is $\lambda_{{\rm min}}(C)=0.26$.
However, 
\[
A^{*}=\left[\begin{array}{ccc}
0.70 & 0.58 & 0.54\\
0.58 & 0.63 & 0.19\\
0.54 & 0.19 & 0.71
\end{array}\right],
\]
has a negative eigenvalue, $\lambda_{\min}(A^{\ast})=-0.02$, and
$C$ can therefore not be expresses as a CF model. It seems unjustified
that the Clustered Factor (CF) model precludes the correlation matrix
in (\ref{eq:3x3block}) (and similar matrices) in advance, because
there does not appear to be anything bizarre or unusual about this
particular block correlation matrix.

Another question is how the number of factors, $r$, are needed to
to generate a given correlation matrix. We address this in the following
Theorem.
\begin{thm}[Minimal factor model]
If $C$ admits a CF model, then the minimal number of factors is
$r={\rm rank}\left(A^{\ast}\right)$. 
\end{thm}
\begin{proof}
If $C$ is generated by a factor model with $\tilde{r}$ factors,
then $A^{*}=\tilde{B}\tilde{B}^{\prime}$ where $\tilde{B}\in\mathbb{R}^{K\times\tilde{r}}$.
We have $r=\mathrm{rank}(A^{\ast})\leq\min(K,\tilde{r})$, and we
can express $A^{*}=BB^{\prime}$ where ${\rm rank}\left(B\right)=r$,
which shows that $C$ admits a factor structure with $r$ orthogonal
factors.
\end{proof}
\medskip{}

Theorem 2 shows that increasing the number of factors beyond $K$
does not broaden the range of attainable correlation matrices. The
(maximum) number of free parameters in the CF model is $K(K+1)/2$,
because we can always rotate the latent factors, $f$, such that $B$
is upper (or lower) triangular. 

\subsection{Testing if $C$ Admits a Factor Representation}

Theorem \ref{thm:FactorRepresentation} shows that the block correlation
matrices that admits a CF representation, are characterized the smallest
eigenvalue of $A^{\ast}$ being nonnegative, i.e. $\lambda_{\min}(A^{\ast})\geq0$.
Here $A^{\ast}$ is defined by (\ref{eq:Astar}) where we set $\rho_{kk}=1$
whenever $n_{k}=1$. 

In practice, it is therefore relatively straight forward to test the
CF representation using the null hypothesis:
\[
H_{F}:\lambda_{\min}(A^{\ast})\geq0.
\]

\subsection{Determining the Minimal Number of Factors}

Given a block correlation matrix with a CF structure, we can proceed
to estimate the minimum number of factors, $r$, as defined by the
rank of $A^{\ast}$. This could be done with the sequential procedure
propose in \citet{Pantula:1989}, which is commonly used to determine
the cointegration rank in VAR models, see \citet[1991]{Johansen88}.\footnote{The same testing principle is used for lag-length selection in time
series model, see \citet{NgPerron01}, and in multiple comparisons
to determine the model confidence set, see \citet{HansenLundeNasonMCS}.}\nocite{Johansen91} 

Let $\lambda_{1}\geq\lambda_{2}\geq\cdots\geq\lambda_{K}$ denote
the eigenvalues of $A^{\ast}$, and consider the hypothesis $H_{F,q}:\lambda_{r+1}=\cdots=\lambda_{K}=0$,
implying that $A^{\ast}$ is psd with $\mathrm{rank}(A^{\ast})\leq r$.
The sequential procedure begins by testing $H_{F,0}$ against $H_{F,K}$.
In the event of a rejection, we proceed to test $H_{F,1}$ against
$H_{F,K}$, and so forth until a null hypothesis is not rejected.
We can set $\hat{r}=r^{\ast}$, where $r^{\ast}$ is the first instance
(the smallest $r$) where $H_{F,r}$ is not rejected. The degrees
of freedom in $B$ for a model with $r\leq K$ factors is $r[K+(K-r)+1]/2$.

\section{A Better Parametrization Block Correlation Matrices }

A new parametrization of correlation matrices was proposed in \citet{ArchakovHansen:Correlation},
and is defined by
\begin{equation}
\gamma={\rm vecl}\left(\log C\right)\in\mathbb{R}^{d},\qquad d=n(n-1)/2,\label{eq:LogCtrans}
\end{equation}
where ${\rm vecl}(\cdot)$ extracts and vectorizes the elements below
the diagonal and $\log C$ is the matrix logarithm of the correlation
matrix.\footnote{For a nonsingular correlation matrix, we have $\log C=Q\log\Lambda Q^{\prime}$,
where $C=Q\Lambda Q^{\prime}$ is the spectral decomposition of $C$,
so that $\Lambda$ is a diagonal matrix with the eigenvalues of $C$.} The identity (\ref{eq:LogCtrans}) defines a one-to-one mapping between
$\mathbb{R}^{d}$ and the set of non-singular $n\times n$ correlation
matrices. 

Because the matrix logarithm preserves the block structure in $C$,
see \citet{ArchakovHansen:CanonicalBlockMatrix}, $\gamma$ will contain
many ``duplicates''. Thus, we can parametrize block correlation
matrices with $\eta\in\mathbb{R}^{q}$, where $\eta$ is a subvector
of $\gamma$ and $q\leq K(K+1)/2$. The follow example will serve
as an illustration,{\small{}
\[
\ensuremath{\underbrace{\left[\begin{array}{cccccc}
1.0 & 0.4 & 0.4 & \cellcolor{gray!15}0.3 & \cellcolor{gray!15}0.3 & \cellcolor{gray!15}0.3\\
0.4 & 1.0 & 0.4 & \cellcolor{gray!15}0.3 & \cellcolor{gray!15}0.3 & \cellcolor{gray!15}0.3\\
0.4 & 0.4 & 1.0 & \cellcolor{gray!15}0.3 & \cellcolor{gray!15}0.3 & \cellcolor{gray!15}0.3\\
\cellcolor{gray!15}0.3 & \cellcolor{gray!15}0.3 & \cellcolor{gray!15}0.3 & 1.0 & 0.6 & 0.6\\
\cellcolor{gray!15}0.3 & \cellcolor{gray!15}0.3 & \cellcolor{gray!15}0.3 & 0.6 & 1.0 & 0.6\\
\cellcolor{gray!15}0.3 & \cellcolor{gray!15}0.3 & \cellcolor{gray!15}0.3 & 0.6 & 0.6 & 1.0
\end{array}\right]}_{=C}\quad\underbrace{\left[\begin{array}{rrrrrr}
-.19 & .326 & .326 & \cellcolor{gray!15}.162 & \cellcolor{gray!15}.162 & \cellcolor{gray!15}.162\\
.326 & -.19 & .326 & \cellcolor{gray!15}.162 & \cellcolor{gray!15}.162 & \cellcolor{gray!15}.162\\
.326 & .326 & -.19 & \cellcolor{gray!15}.162 & \cellcolor{gray!15}.162 & \cellcolor{gray!15}.162\\
\cellcolor{gray!15}.162 & \cellcolor{gray!15}.162 & \cellcolor{gray!15}.162 & -.38 & .533 & .533\\
\cellcolor{gray!15}.162 & \cellcolor{gray!15}.162 & \cellcolor{gray!15}.162 & .533 & -.38 & .533\\
\cellcolor{gray!15}.162 & \cellcolor{gray!15}.162 & \cellcolor{gray!15}.162 & .533 & .533 & -.38
\end{array}\right]}_{=\log C}}.
\]
}Here we can use $\eta=(\begin{array}{ccc}
0.326 & \cellcolor{gray!15}0.162 & 0.533\end{array})^{\prime}$ as the condensed vector parametrization of the block correlation
matrix. For a given block partitioning, $(n_{1},\ldots,n_{K})$, any
non-singular block correlation matrix will map to a unique vector
$\eta$, and any vector $\eta\in\mathbb{R}^{q}$ will map to a unique
non-singular block correlation matrix, see \citet{ArchakovHansen:CanonicalBlockMatrix}. 

It is straightforward to add additional structure onto the $\eta$-parametrization,
for instance by restricting $\eta$ to be in a subspace of lower dimension
than $q$. For a multivariate GARCH model, these ideas are explored
in \citet{ArchakovHansenLundeMRG}.

\citet{CrealKim:BayesianBlockCorr} recently adopted this parametrization
for Bayesian modeling of block correlation matrices. An attractive
feature of this parametrization, it that it facilitates priors with
full support on the entire set of non-singular block correlation matrices.

\section{Summary\label{sec:Summery}}

We have characterized the class of block correlation matrices that
can be expressed as a clustered factor model. While the clustered
factor model serves as a valuable tool for generating clustered correlation
structures, it does introduce unnecessary constraints on the correlation
matrix, which may limit its practical relevance for some empirical
problems. The alternative parametrization, which is based on the matrix
logarithm of the correlation matrix, seems better suited for the modeling
of block correlation matrices. It avoids the imposition of superfluous
constraints on the correlation matrix and, if needed, it provides
a flexible framework for further reduction of the degrees of freedom.

\bibliographystyle{apalike}
\bibliography{prh}

\end{document}